\documentclass{llncs}
\usepackage{epsfig}
\usepackage{amsmath}
\usepackage{amssymb}
\usepackage{algorithmic}
\usepackage{algorithm}
\usepackage{bigstrut}

\def\idtt#1{\ensuremath{\mathtt{#1}}}

\def\rankop{\idtt{rank}}

\def\accessop{\idtt{access}}

\title{Succinct Posets}

\author{J. Ian Munro and Patrick K. Nicholson\thanks{This research was
    funded in part by NSERC of Canada, and the Canada Research Chairs
    program.}}

\institute{David R. Cherition School of
  Computer Science, University of Waterloo}
\begin{document}
\pagestyle{plain}

\maketitle

\begin{abstract}
We describe an algorithm for compressing a partially ordered set, or
\emph{poset}, so that it occupies space matching the information
theory lower bound (to within lower order terms), in the worst case.
Using this algorithm, we design a succinct data structure for
representing a poset that, given two elements, can report whether one
precedes the other in constant time.  This is equivalent to succinctly
representing the transitive closure graph of the poset, and we note
that the same method can also be used to succinctly represent the
transitive reduction graph.  For an $n$ element poset, the data
structure occupies $n^2/4 + o(n^2)$ bits, in the worst case, which is
roughly half the space occupied by an upper triangular matrix.
Furthermore, a slight extension to this data structure yields a
succinct oracle for reachability in arbitrary directed graphs.  Thus,
using roughly a quarter of the space required to represent an
arbitrary directed graph, reachability queries can be supported in
constant time.
\end{abstract}

\section{Introduction}

Partially ordered sets, or \emph{posets}, are useful for modelling
relationships between objects, and appear in many different areas,
such as natural language processing, machine learning, and database
systems.  As problem instances in these areas are ever-increasing in
size, developing more space efficient data structures for representing
posets is becoming an increasingly important problem.

When designing a data structure to represent a particular type of
combinatorial object, it is useful to first determine how many objects
there are of that type.  By a constructive enumeration argument,
Kleitman and Rothschild~\cite{KR75} showed that the number of $n$
element posets is $2^{n^2/4 + O(n)}$.  Thus, the information theoretic
lower bound indicates that representing an arbitrary poset requires
$\lg(2^{n^2/4+O(n)}) = n^2/4 + O(n)$ bits\footnote{We use $\lg n$ to
  denote $\lceil \log_2 n \rceil$.}.  This naturally raises the
question of how a poset can be represented using only $n^2/4 + o(n^2)$
bits, \emph{and} support efficient query operations.  Such a
representation, that occupies space matching the information theoretic
lower bound to within lower order terms while supporting efficient
query operations, is called a \emph{succinct data
  structure}~\cite{J89}.

The purpose of this paper is to answer this question by describing the
first succinct representation of arbitrary posets.  We give a detailed
description of our results in Section~\ref{sec:results}, but first
provide some definitions in Section~\ref{sec:def} and then highlight
some of the previous work related to this problem in
Section~\ref{sec:prev}.

\section{Definitions\label{sec:def}}

A poset $P$, is a reflexive, antisymmetric, transitive binary relation
$\preceq$ on a set of $n$ elements $S$, denoted $P=(S,\preceq)$.  Let
$a$ and $b$ be two elements in $S$.  If $a \preceq b$, we say $a$
\emph{precedes} $b$.  We refer to queries of the form, ``Does $a$
precede $b$?'' as \emph{precedence queries}.  If neither $a \preceq b$
or $b \preceq a$, then we say $a$ and $b$ are \emph{incomparable}. For
convenience we write $a \prec b$ if $a \preceq b$ and $a \neq b$.

Each poset $P = (S,\preceq)$ is uniquely described by a directed
acyclic graph, or \emph{DAG}, $G_c = (S,E_c)$, where $E_c = \{(a,b) :
a \prec b\}$ is the set of edges.  The DAG $G_c$ is the
\emph{transitive closure graph} of $P$.  Note that a precedence query
for elements $a$ and $b$ is equivalent to the query, ``Is the edge
$(a,b)$ in $E_c$?''  Alternatively, let $G_r = (S,E_r)$ be the DAG
such that $E_r = \{(a,b) : a \prec b, \nexists_{c \in S}, a \prec c
\prec b \}$, i.e., the minimal set of edges that imply all the edges
in $E_c$ by transitivity.  The DAG $G_r$ also uniquely describes $P$,
and is called the \emph{transitive reduction graph} of $P$.  

Posets are also sometimes illustrated using a \emph{Hasse diagram},
which displays all the edges in the transitive reduction, and
indicates the direction of an edge $(a,b)$ by drawing element $a$
above $b$.  We refer to elements that have no outward edges in the
transitive reduction as \emph{sinks}, and elements that have no inward
edges in the transitive reduction as \emph{sources}.  See
Figure~\ref{fig:red-clos} for an example.  Since all these concepts
are equivalent, we may freely move between them when discussing a
poset, depending on which representation is the most convenient.

\begin{figure}
\centering
\includegraphics{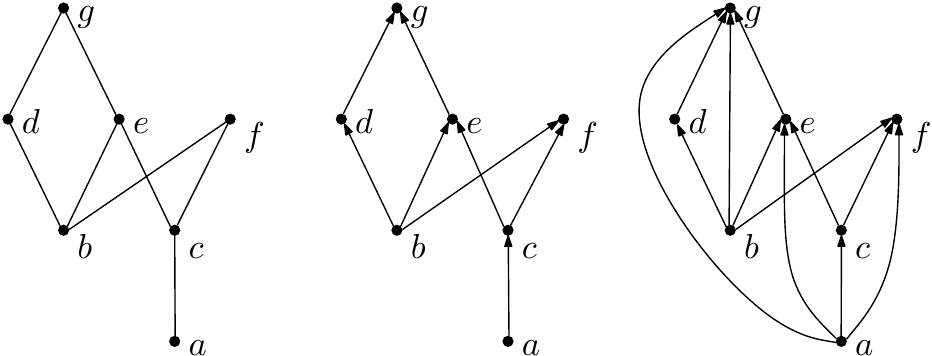}
\caption{\label{fig:red-clos}A Hasse diagram of a poset (left), the
  transitive reduction (centre), and the transitive closure
  (right). Elements $a$ and $b$ are sources, and elements $g$ and $f$
  are sinks.}
\end{figure}

A \emph{linear extension} $L = \{a_1, ... ,a_n\}$ is a total ordering
of the elements in $S$ such if $a_i \prec a_j$ for some $i \neq j$,
then $i < j$.  However, note that the converse is not necessarily
true: we cannot determine whether $a_i \prec a_j$ unless we know that
$a_i$ and $a_j$ are comparable elements.  A \emph{chain} of a poset,
$P = (S,\preceq)$, is a total ordering $C = \{a_1, ... ,a_k\}$ on a
subset of $k$ elements from $S$ such that $a_i \prec a_j$ iff $i < j$,
for $1 \le i < j \le k$.  An \emph{antichain} is a set $A = \{a_1 ,
... , a_k \}$ that is a subset of $k$ elements from $S$, such that
each $a_i$ and $a_j$ are incomparable, for $1 \le i < j \le k$.  The
\emph{height} of a poset is the size of its maximum length chain, and
the \emph{width} of a poset is the size of its maximum antichain.

For a graph $G = (V,E)$, we sometimes use $E(H)$ to denote the set of
edges $\{(a,b) : (a,b) \in E, a \in H, b \in H \}$, where $H \subseteq
V$.  Similarly, we use $G(H)$ to denote the subgraph of $G$ induced by
$H$, i.e., the subgraph with vertex set $H$ and edge set $E(H)$.
Finally, if $(a,b) \in E$, or $(b,a) \in E$, we say that $b$ is a
\emph{neighbour} of $a$ in $G$.

\section{Previous work\label{sec:prev}}

Previous work in the area of succinct data structures includes
representations of arbitrary undirected graphs~\cite{AM08}, planar
graphs~\cite{AHM12}, and trees~\cite{MR01}. There has also been
interest in developing reachability oracles for planar directed
graphs~\cite{T04}, as well as approximate distance oracles for
undirected graphs~\cite{TZ05}.  For restricted classes of posets, such
as lattices~\cite{TV99} and distributive lattices~\cite{HN96}, space
efficient representations have been developed, though they are not
succinct.

One way of storing a poset is by representing either its transitive
closure graph, or transitive reduction graph, using an adjacency
matrix.  If we topologically order the vertices of this graph, then we
can use an upper triangular matrix to represent the edges, since the
graph is a DAG.  Such a representation occupies $\binom{n}{2}$ bits,
and can, in a single bit probe, be used to report whether an edge
exists in the graph between two specified elements.  Thus, using this
simple approach we can achieve a space bound that is roughly two times
the information theory lower bound for representing a poset.  An
alternative representation, called the \emph{ChainMerge} structure was
proposed by Daskalakis et al.~\cite{DKMRV09}, that occupies $O(nw)$
\emph{words} of space, where $w$ is the width of the poset.  The
ChainMerge structure, like the transitive closure graph, supports
precedence queries in $O(1)$ time.

Recently, Farzan and Fischer~\cite{FF11} presented a data structure
that represents a poset using $2nw(1+o(1)) + (1+\varepsilon)n \lg n$
bits, where $w$ is the width of the poset, and $\varepsilon > 0$ is an
arbitrary positive constant.  This data structure supports precedence
queries in $O(1)$ time, and many other operations in time proportional
to the width of the poset.  These operations are best expressed in
terms of the transitive closure and reduction graphs, and include:
reporting all neighbours of an element in the transitive closure in
$O(w + k)$ time, where $k$ is the number of reported elements;
reporting all neighbours of an element in the transitive reduction in
$O(w^2)$ time; reporting an arbitrary neighbour of an element in the
transitive reduction in $O(w)$ time; reporting whether an edge exists
between two elements in the transitive reduction in $O(w)$ time;
reporting all elements that, for two elements $a$ and $b$, are both
preceded by $a$ and precede $b$ in $O(w + k)$ time; among others.  The
basic idea of their data structure is to encode the ChainMerge
structure of Daskalakis et al.~\cite{DKMRV09} using bit sequences, and
answer queries using rank and select operations on these bit
sequences.

Since the data structure of Farzan and Fischer~\cite{FF11} is adaptive
on width, it is appropriate for posets where the width is a
slow-growing function of $n$.  However, if we select a poset of $n$
elements uniformly at random from the set of all possible $n$ element
posets, then it will have width $n/2 + o(n)$ with high
probability~\cite{KR75}.  Thus, this representation may occupy as many
as $n^2 + o(n^2)$ bits, which is roughly \emph{four times} the
information theory lower bound.  Furthermore, with the exception of
precedence queries, all other operations take linear time for such a
poset.

\section{Our Results\label{sec:results}}
Our results hold in the word-RAM model of computation with word size
$\Theta(\lg n)$ bits.  Our main result is summarized in the following
theorem:

\begin{theorem}\label{thm:main}
Let $P = (S, \preceq)$ be a poset, where $|S| = n$.  There is a
succinct data structure for representing $P$ that occupies $n^2/4 +
O(( n^2 \lg \lg n) / \lg n)$ bits, and can support precedence queries
in $O(1)$ time: i.e., given two elements $a, b \in S$, report whether
$a \preceq b$.
\end{theorem}

The previous theorem implies that we can, in $O(1)$ time, answer
queries of the form, ``Is the edge $(a,b)$ in the transitive closure
graph of $P$?''  In fact, we can also apply the same representation to
support, in $O(1)$ time, queries of the form, ``Is the edge $(a,b)$ in
the transitive reduction graph of $P$?'' However, at present it seems
as though we can only support efficient queries in one or the other,
\emph{not both} simultaneously.  For this reason we focus on the
closure, since it is likely more useful, but state the following
theorem:

\begin{theorem}\label{thm:tr}
Let $G_r = (S,E_r)$ be the transitive reduction graph of a poset,
where $|S| = n$.  There is a succinct data structure for representing
$G_r$ that occupies $n^2/4 + O(( n^2 \lg \lg n) / \lg n)$ bits, and,
given two elements $a,b \in S$, can report whether $(a,b) \in E_r$ in
$O(1)$ time.
\end{theorem}

\paragraph{Reachability in Directed Graphs}: 
For an arbitrary DAG, the \emph{reachability relation} between
vertices is a poset: i.e., given two vertices, $a$ and $b$, the
relation of whether there a directed path from $a$ to $b$ in the DAG.
As a consequence, Theorem~\ref{thm:main} implies that there is a data
structure that occupies $n^2/4 + o(n^2)$ bits, and can support
reachability queries in a DAG, in $O(1)$ time. We can even strengthen
this observation by noting that for an \emph{arbitrary directed graph}
$G$, the \emph{condensation} of $G$--- the graph that results by
contracting each strongly connected component into a single
vertex~\cite[Section 22.5]{CLRS01}--- is a DAG.  Given two vertices
$a$ and $b$, if $a$ and $b$ are in the same strongly connected
component, then $b$ is reachable from $a$.  Otherwise, we can apply
Theorem~\ref{thm:main} to the condensation of $G$.  Thus, we get the
following corollary:

\begin{corollary}\label{cor:digraph}
Let $G$ be a directed graph.  There is a data structure that occupies
$n^2/4 + o(n^2)$ bits and, given two vertices of $G$, $a$ and $b$, can
report whether $b$ is reachable from $a$ in $O(1)$ time.
\end{corollary}

\noindent
Note that the space bound of the previous corollary is roughly a
quarter of the space required to represent an arbitrary directed
graph!  Switching back to the terminology of order theory, the
previous corollary generalizes Theorem~\ref{thm:main} to the larger
class of binary relations known as \emph{quasi-orders}: i.e., binary
relations that are reflexive and transitive, but not necessarily
antisymmetric.  In fact, reflexivity does not restrict the binary
relation very much, so we can further generalize
Theorem~\ref{thm:main} to arbitrary \emph{transitive binary
  relations}; we discuss this in Appendix~\ref{app:transitive}.

\paragraph{Overview of the data structure:}
The main idea behind our succinct data structure is to develop an
algorithm for compressing a poset so that it occupies space matching
the information theory lower bound (to within lower order terms), in
the worst case.  The main difficulty is ensuring that we are able to
query the compressed structure efficiently.  Our first attempt at
designing a compression algorithm was essentially a reverse engineered
version of an enumeration proof by Kleitman and
Rothschild~\cite{KR70}.  However, though the algorithm achieved the
desired space bound, there was no obvious way to answer queries on the
compressed data due to one crucial compression step.  Though there are
several other enumeration proofs (cf., ~\cite{KR75,BPS96}), they all
appeal to a similar strategy, making the compressed data difficult to
query.  This led us to develop an alternate compression algorithm,
that uses techniques from extremal graph theory.

We believe it is conceptually simpler to present our algorithm as
having two steps.  In the first step, we preprocess the poset,
removing edges in its transitive closure graph, to create a new poset
where the height is not too large.  We refer to what remains as a
\emph{flat} poset.  We then make use of the fact that, in a flat
poset, either balanced biclique subgraphs of the transitive closure
graph--- containing $\Omega(\lg n/\lg \lg n)$ elements--- must exist,
or the poset is relatively sparsely connected.  In the former case,
the connectivity between these balanced biclique subgraphs and the
remaining elements is shown to be space efficient to encode using the
fact that all edges implied by transitivity are in the transitive
closure graph.  In the latter case, we can directly apply techniques
from the area of succinct data structures to compress the poset.

\section{\label{sec:ds}Succinct Data Structure}

In this section we describe a succinct data structure for representing
posets.  In order to refer to the elements in the poset, we assume
each element has a label.  Since our goal is to design a data
structure that occupies $n^2/4 + o(n^2)$ bits, we are free to assign
arbitrary $O(\lg n)$-bit labels to the elements, as such a labeling
will require only $O(n \lg n)$ bits.  Thus, we can assume each element
in our poset has a distinct integer label, drawn from the range
$[1,n]$.  Our data structure always refers to elements by their
labels, so often when we refer to ``element'' $a$, it means ``the
element in $S$ with label $a$'', depending on context.

\subsection{Preliminary Data Structures}

Given a bit sequence $B[1..n]$, we use $\accessop(B,i)$ to denote the
$i$-th bit in $B$, and $\rankop(S,i)$ to denote the number of 1 bits
in the prefix $B[1..i]$.  We make use of the following lemma, which
can be used to support access and rank operations on bit sequences,
while compressing the sequence to its 0th-order empirical entropy.

\begin{lemma}[Raman, Raman, Rao~\cite{RRR02}]\label{lem:rrr}
Given a bit sequence $B$ of length $n$, of which $\beta$ bits are $1$,
there is a data structure that can represent $B$ using $\lg
\binom{n}{\beta} + O(n\lg \lg n / \lg n)$ bits that can support the
operations $\accessop$, and $\rankop$ on $B$ in $O(1)$ time.
\end{lemma}

\subsection{Flattening a Poset}

Let $\gamma > 0$ be a parameter, to be fixed later; the reader would
not be misled by thinking that we will eventually set $\gamma = \lg
n$.  We call a poset $\gamma$-\emph{flat} if it has height no greater
than $\gamma$.  In this section, we describe a preprocessing algorithm
for posets that outputs a data structure of size $O(n^2/\gamma)$ bits,
that transforms a poset into a $\gamma$-flat poset, without losing any
information about its original structure.  After describing this
preprocessing algorithm, we develop a compression algorithm for flat
posets.  Using the preprocessing algorithm together with the
compression algorithm yields a succinct data structure for posets.

Let $P = (S,\preceq)$ be an arbitrary poset with transitive closure
graph $G_c = (S,E_c)$. We decompose the elements of $S$ into
antichains based on their height within $P$.  Let $\mathcal{H}(P)$
denote the height of $P$.  All the sources in $S$ are of height $1$,
and therefore are assigned to the same set.  Each non-source element
$a \in S$ is assigned a height equal to the length of the maximum path
from a source to $a$.  We use $U_h$ to denote the set of all the
elements of height $h$, $1 \le h \le \mathcal{H}(P)$, and
$\mathcal{U}$ to denote the set $\{U_1, ..., U_{\mathcal{H}(P)}\}$.
Furthermore, it is clear that each set, $U_h$, is an antichain, since
if $a \prec b$ then the height of $b$ is strictly greater than $a$.

Next, we compute a linear extension $\mathcal{L}$ of the poset $P$ in
the following way, using $\mathcal{U}$.  The linear extension
$\mathcal{L}$ is ordered such that all elements in $U_i$ come before
$U_{i+1}$ for all $1 \le i < \mathcal{H}(P)$, and the elements within
the same $U_i$ are ordered arbitrarily within $\mathcal{L}$.  Given
any subset $S' \subseteq S$, we use the notation $S'(x)$ to denote the
element ranked $x$-th according to $\mathcal{L}$, among the elements
in the subset $S'$.  We illustrate these concepts in
Figure~\ref{fig:antichain}.  Later, this particular linear extension
will be used extensively, when we output the structure of the poset as
a bit sequence.

\begin{figure}
\centering
\includegraphics{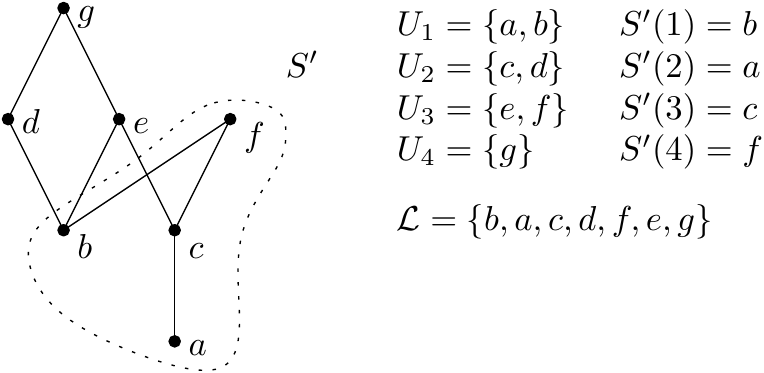}
\caption{\label{fig:antichain}The antichain decomposition of the
  poset from Figure~\ref{fig:red-clos}.  The set $S'$ is the
  set of elements surrounded by the dotted line.  Note that
  $\mathcal{L}$ is only one of many possible linear extensions.}
\end{figure}

We now describe a preprocessing algorithm to transform an arbitrary
poset $P$ into a $\gamma$-flat poset $\tilde{P}$.  We assume $P$ is
not $\gamma$-flat, otherwise we are done. Given two consecutive
antichains $U_i$ and $U_{i+1}$, we define a \emph{merge step} to be
the operation of replacing $U_{i}$ and $U_{i+1}$ by a new antichain
$U_{i}' = U_i \cup U_{i+1}$, and outputting and removing all the edges
between elements in $U_i$ and $U_{i+1}$ in the transitive closure of
$P$, i.e., $E_c(U_i \cup U_{i+1})$. We say that $U_{i+1}$ is the
\emph{upper antichain}, $U_i$ is the \emph{lower antichain}, and refer
to the new antichain $U_{i}' $ as the \emph{merged antichain}.  Each
antichain $U_j$ where $j > i+1$ becomes antichain $U_{j-1}'$ in the
\emph{residual decomposition}, after the merge step.  To represent the
edges, let $B$ be a bit sequence, storing $|U_i||U_{i+1}|$ bits.  The
bit sequence $B$ is further subdivided into sections, denoted $B^x$,
for each $x \in [1,|U_i|]$, where the bit $B^{x}[y]$ represents
whether there is an edge from $U_i(x)$ to $U_{i+1}(y)$; or
equivalently, whether $U_i(x) \prec U_{i+1}(y)$.  We say that
antichain $U_{i+1}$ is \emph{associated} with $B$, and vice versa.
The binary string $B$ is represented using the data structure of
Lemma~\ref{lem:rrr}, which compresses it to its 0th-order empirical
entropy\footnote{We note that for our purposes in this section,
  compression of the bit sequence is not required to achieve the
  desired asymptotic space bounds.  However, the fact that
  Lemma~\ref{lem:rrr} compresses the bit sequence will indeed matter
  in Section~\ref{sec:comp-flat}.}.  Note that, after the merge step,
the elements in merged antichain $U_i'$ are ordered, in the linear
extension $\mathcal{L}$, such that $U_i'(x) = U_i(x)$ for $1 \le x \le
|U_i|$ and $U_i'(y + |U_i|) =U_{i+1}(y)$ for $1 \le y \le |U_{i+1}|$.

\begin{algorithm}
\renewcommand{\thealgorithm}{}
\caption{$\textsc{Flatten}(\mathcal{U},i)$: where $i$ is the
  index of an antichain in $\mathcal{U}$.}
\begin{algorithmic}
\IF{$i > |\mathcal{U}|$} 
  \STATE\textsc{Exit}
\ENDIF
\IF{$|U_i| + |U_{i+1}| \le 2n/\gamma$}
  \STATE Perform a merge step on $U_i$ and $U_{i+1}$
\ELSE
  \STATE $i \leftarrow i + 1$
\ENDIF
\STATE $\textsc{Flatten}(\mathcal{U},i)$
\end{algorithmic}
\end{algorithm}

There are many possible ways that we could apply merge steps to the
poset in order to make it $\gamma$-flat.  The method we choose,
presented in algorithm $\textsc{Flatten}$, has the added benefit that
accessing the output bit sequences is straightforward.  Let
$\tilde{\mathcal{U}}$ be the residual antichain decomposition that
remains after executing $\textsc{Flatten}(\mathcal{U},1)$, and
$\tilde{P}$ be the resulting poset.  The number of antichains in
$\tilde{\mathcal{U}}$ is at most $\gamma$, and therefore the remaining
poset $\tilde{P}$ is $\gamma$-flat.  We make the following further
observation:

\begin{lemma}\label{lem:flatten}
$\textsc{Flatten}(\mathcal{U},1)$ outputs $O(n^2/\gamma)$ bits.
\end{lemma}

\begin{proof}
Consider the decomposition $\mathcal{U}$ and let $m = \mathcal{H}(P) =
|\mathcal{U}|$.  Let $n_1, ..., n_m$ denote the number of elements in
$U_1, ..., U_m$, and $n_{s,t}$ to denote $\sum_{i = s}^t n_i$. We use
the fact that the expression $\sum_{i = s}^{t - 1} ((\sum_{j = s}^{i}
n_j) n_{i+1}) \le n_{s,t}(n_{s,t} -1) /2$, where $1 \le s < t \le m$;
we include a proof in Appendix~\ref{app:inequality}.  For each of the
at most $\gamma$ antichains in $\tilde{\mathcal{U}}$, the previous
inequality implies that \textsc{Flatten} outputs no more than
$O(n_{s,t}^2)$ bits, where $n_{s,t} = O(n/\gamma)$.  Thus, overall the
number of bits output during the merging steps is $O((n /\gamma)^2
\gamma) = O(n^2 /\gamma)$. \qed
\end{proof}

We now show how to use the output of the merge steps to answer
connectivity queries for edges that were removed by the
$\textsc{Flatten}$ algorithm:

\begin{lemma}\label{lem:flatten-index}
There is a data structure of size $O(n^2/\gamma)$ bits that, given two
elements $a$ and $b$ can determine in $O(1)$ time whether $a$ precedes
$b$, if both $a$ and $b$ belong to the same antichain in the residual
antichain decomposition $\tilde{\mathcal{U}}$.
\end{lemma}

\begin{proof}
We add additional data structures to the output of $\textsc{Flatten}$
in order to support queries.  Since the labels of elements in $S$ are
in the range $[1,n]$, we can treat elements as array indices.  Thus,
it is trivial to construct an $O(n \lg n)$ bit array that, given
elements $a,b \in S$, returns values $i,i',j,j',x,x',y$ and $y'$ in
$O(1)$ time such that $U_i(x) = a$, $U_j(y) = b$, $U_{i'}(x') = a$,
$U_{j'}(y') = b$, where $U_i, U_j \in \mathcal{U}$ and $U_{i'},U_{j'}
\in \tilde{\mathcal{U}}$.  We also store an array $A$ containing
$|\mathcal{U}|$ \emph{records}.  For each antichain $U_i \in
\mathcal{U}$, if $U_i$ is the upper antichain during a merge
step\footnote{Note that, with the exception of the first merge step,
  $U_i \in \mathcal{U}$ is \emph{not} the $i$-th antichain in the
  decomposition when the merge step occurs, but we will store records
  for the index $i$ rather than some intermediate index.}, then:
$A[i].\texttt{pnt}$ points to the start of the sequence, $B$,
associated with $U_i$, and; $A[i].\texttt{len}$ stores the length of
the lower antichain. Recall that after the merge step, the element
$U_i(x)$ has rank $x + A[i].\texttt{len}$ in the merged antichain.
Thus, $A[i].\texttt{len}$ is the \emph{offset} of the ranks of the
elements of $U_i$ within the merged antichain.  These extra data
structures occupy $O(n\lg n)$ bits and are dominated by the size of
the output of $\textsc{Flatten}$, so the claimed space bound holds by
Lemma~\ref{lem:flatten}.

We now discuss how to answer a query.  Given $a,b \in S$, if $i' \neq
j'$, then we return ``different antichains''.  Otherwise, if $i = j$,
then we return ``no''.  Otherwise, assume without loss of generality
that $i > j$.  Thus, $U_i$ is the upper antichain, and
$A[i].\texttt{pnt}$ is a pointer to a sequence $B$, whereas $U_j$ is a
subset of the lower antichain $\hat{U}_k$, and $A[j].\texttt{len}$ is
the offset of the elements in $U_j$ within $\hat{U}_k$.  Let $z = y +
A[j].\texttt{len}$, and return ``yes'' if $B^z[x] = 1$ and ``no''
otherwise. Section $B^z$ begins at the $((z-1)|U_i|)$-th bit of $B$ so
we can access $B^z[x]$ in $O(1)$ time.  \qed
\end{proof}

\subsection{Compressing Flat Posets\label{sec:comp-flat}}

In this section we describe a compression algorithm for flat posets
that, in the worst case, matches the information theory lower bound to
within lower order terms.  We begin by stating the following lemma,
which is a constructive deterministic version of a well known theorem
by K\"{o}v\'{a}ri, S\'{o}s, and Tur\'{a}n~\cite{KST1954}:

\begin{lemma}[Mubayi and Tur\'{a}n~\cite{MT2010}]\label{lem:kst}
There is a constant $c_\text{min}$ such that, given a graph with $|V|
\ge c_\text{min}$ vertices and $|E| \ge 8|V|^{3/2}$ edges, we can find
a balanced biclique $K_{q,q}$, where $q = \Theta(\lg |V| / \lg (|V|^2
/ |E|))$, in time $O(|E|)$.
\end{lemma}

Let $\tilde{P}$ be a $(\lg n)$-flat poset, $G_c = (S,E_c)$ be its
transitive closure, and $\tilde{\mathcal{U}} = \{U_1, ..., U_m\}$ be
its antichain decomposition (discussed in the last section), which
contains $m \le \lg n$ antichains.  We now prove our key lemma,
which is crucial for the compression algorithm.

\begin{lemma}[Key Lemma]\label{lem:clos-connect}
Consider the subgraph $G_{\Upsilon} = G_c(U_i \cup U_{i+1})$ for some
$1 \le i < m$, and ignore the edge directions so that $G_{\Upsilon}$ is
undirected.  Suppose $G_{\Upsilon}$ contains a balanced biclique
subgraph with vertex set $D$, and $|D| = \tau$.  Then there are at
most $2^{\tau/2 +1} -1$ ways that the vertices in $D$ can be connected
to each vertex in $S \setminus (U_i \cup U_{i+1})$.
\end{lemma}

\begin{proof}
Each vertex $v \in S \setminus (U_i \cup U_{i+1})$ is in $U_j$, where,
either $j > i+1$ or $j < i$.  Without loss of generality, consider the
case where $j > i+1$.  If $v$ is connected to any vertex $u \in D\cap
U_{i+1}$, then $v$ is connected to \emph{all} vertices in $D \cap
U_i$.  Thus, $v$ can be connected to the vertices in $D \cap U_{i+1}$
in $2^{\tau/2}-1$ ways, or to the vertices in $D \cap U_i$ in
$2^{\tau/2}-1$ ways, or not connected to $D$ at all.  In total, there
are $2^{\tau/2+1} - 1$ ways to connect $v$ to $D$. \qed
\end{proof}

\begin{algorithm}
\renewcommand{\thealgorithm}{}
\caption{$\textsc{Compress-Flat}(\hat{P},\hat{n},\hat{\mathcal{U}},\hat{m})$:
  where $\hat{P} = (\hat{S},\preceq)$ is a $(\lg n)$-flat poset of
  $\hat{n} \le n$ elements, and $\hat{\mathcal{U}} = \{\hat{U}_1, ...,
  \hat{U}_{\hat{m}}\}$ is a decomposition of the elements in $\hat{P}$
  into $\hat{m}$ antichains.}
\begin{algorithmic}[1]

\IF{$\hat{m} = 1$} 

 \STATE \label{line:bc1}EXIT

\ELSIF{\label{line:dense-cond}$|\hat{U}_i \cup \hat{U}_{i+1}| \ge
  c_\text{min}$ and $|E_c(\hat{U}_i \cup \hat{U}_{i+1})| \ge (\hat{n}
  / \lg \hat{n})^2$, for an $i \in [1,\hat{m}]$}

 \STATE \label{line:dense}Apply Lemma~\ref{lem:kst} to the subgraph
 $G_c(\hat{U}_i \cup \hat{U}_{i+1})$.  This computes a balanced
 biclique with vertex set $D \subset \hat{U}_i \cup \hat{U}_{i+1}$
 such that $\tau = |D| = \Omega(\lg \hat{n} / \lg \lg \hat{n})$.
 
 \STATE For each element $b \in \hat{U}_i \cap D$ output a bit
 sequence $W^{-}_b$ of $|\hat{U}_{i+1}|$ bits, where $W^{-}_b[k] = 1$
 iff $b \prec \hat{U}_{i+1}(k)$.
 
 \STATE For each element $a \in \hat{U}_{i+1} \cap D$ output a bit
 sequence $W^{+}_a$ of $|\hat{U}_i|$ bits, where $W^{+}_a[k] = 1$ iff
 $\hat{U}_{i}(k) \prec a$.

 \STATE Let $H = \hat{S} \setminus (\hat{U}_i \cup \hat{U}_{i+1})$. Output an array of
 integers $Y$, where $Y[k] \in [0,2^{\tau/2 +1}-1]$ and indicates how
 $H(k)$ is connected to $D$ (see Lemma~\ref{lem:clos-connect}).

 \STATE Set $\hat{U}_i \leftarrow \hat{U}_i \setminus D$

 \STATE Set $\hat{U}_{i+1} \leftarrow \hat{U}_{i+1} \setminus D$

 \STATE \label{line:dense-end}$\textsc{Compress-Flat}(\hat{P}\setminus D , \hat{n} -\tau,
 \hat{\mathcal{U}},\hat{m})$

\ELSE

 \STATE \label{line:sparse}Perform a merge step on $\hat{U}_{1}$ and $\hat{U}_{2}$

 \STATE Set $\hat{m} \leftarrow \hat{m} -1$ 

 \STATE \label{line:sparse-end}$\textsc{Compress-Flat}(\hat{P},\hat{n},\hat{\mathcal{U}},\hat{m})$

\ENDIF
\end{algorithmic}
\end{algorithm}

Consider the algorithm $\textsc{Compress-Flat}$.  The main idea is to
repeatedly apply Lemma~\ref{lem:kst} to two consecutive antichains the
antichain decomposition that have many edges--- defined on
line~\ref{line:dense-cond}--- between them in the transitive closure
graph.  If no such antichains exist, then we apply merge steps.  The
algorithm terminates when only one antichain remains.  We refer to the
case on lines~\ref{line:dense}-\ref{line:dense-end} as the \emph{dense
  case}, and the case on lines~\ref{line:sparse}-\ref{line:sparse-end}
as the \emph{sparse case}.  We now prove that the size of the output
of the compression algorithm matches the information theory lower
bound to within lower order terms.

\begin{lemma}\label{lem:compress-flat}
The output of $\textsc{Compress-Flat}(\tilde{P},n,\tilde{U},m)$ is no
more than $n^2/4 + O((n^2 \lg \lg n) / \lg n)$ bits.
\end{lemma}

\begin{proof}[Sketch]
In the base case (line~\ref{line:bc1}), the lemma trivially holds
since nothing is output.  Next we give the intuition to show that the
total output from all the sparse cases cannot exceed $O((n^2\lg \lg n)
/ \lg n)$ bits.  Recall that the representation of Lemma~\ref{lem:rrr}
compresses to $\lg \lceil \binom{t}{\beta}\rceil + O(t \lg \lg t/ \lg
t)$ bits, where $t$ is the length of the bit sequence, and $\beta$ is
the number of $1$ bits.  We use the fact that $\lg \lceil
\binom{t}{\beta} \rceil \le \beta \lg (et/\beta) +
O(1)$~\cite[Section~4.6.4]{M07}.  For a single pass through the sparse
case, the total number of bits represented by $B$ is $t=O(n^2)$, and
$\beta=O((n / \lg n)^2)$ bits are $1$'s.  Thus, the first term in the
space bound to represent $B$ using Lemma~\ref{lem:rrr} (applying the
inequality) is $O((n^2 \lg \lg n)/ \lg^2 n)$ bits.  Since we can enter
the sparse case \emph{at most} $\lg n$ times before exiting on
line~\ref{line:bc1}, the total number of bits occupied by the first
term is bounded by $O((n^2 \lg \lg n) / \lg n)$.  To ensure the second
term ($O(t \lg \lg t / \lg t)$) in the space bound of
Lemma~\ref{lem:rrr} does not dominate the cost, we use the standard
technique of applying Lemma~\ref{lem:rrr} to the concatenation of all
the bit sequences output in the sparse case, rather than each
individual sequence separately (see Appendix~\ref{app:compress-flat}
for more details).

We now prove the lemma by induction for the dense case.  Let
$\mathcal{S}(n)$ denote the number of bits output by
$\textsc{Compress-Flat}(\tilde{P},n,\tilde{U},m)$.  \emph{Inductive
  step}: We can assume $\mathcal{S}(n_0) \le n_0^2/4 + c_0 (n_0^{2}
\lg \lg n_0 ) / \lg n_0$ for all $1 \le n_0 < n$, where $n \ge 2$, and
$c_0 > 0$ is some sufficiently large constant.  All the additional
self-delimiting information--- for example, storing the length of the
sequences output on lines 5-7--- occupies no more than $c_1\lg n$ bits
for some constant $c_1 > 0$.  Finally, recall that $\tau \ge c_2 \lg n
/ \lg \lg n$ for some constant $c_2 > 0$.  We have:

\begin{align*}
\mathcal{S}(n) \enspace = \enspace & \frac{\tau}{2}\left(|U_i| +
|U_j|\right) + (n - (|U_i| + |U_j|)) \lg(2^{\tau / 2 +1}) + c_1\lg n
+ \mathcal{S}(n - \tau) \displaybreak[0]\\ 
\le \enspace & (\frac{\tau}{2} +1) n + c_1 \lg n + \frac{1}{4}\left(n^2 -
2n\tau+ \tau^2\right) + \frac{c_0 \lg \lg
  n}{\lg(n-\tau)}\left(n^2 - 2n\tau + \tau^2\right) \displaybreak[0]\\
\le \enspace & c_3 n + \frac{n^2}{4} + \frac{c_0 n^2 \lg \lg
  n}{\lg(n-\tau)} - c_4 n  \text{\enspace \enspace($c_4 < c_0c_2$, $c_3 > 1$)} \displaybreak[0]\\
\le \enspace & \frac{n^2}{4} + \frac{c_0 n^2 \lg \lg
  n}{\lg(n-\tau)} - c_5 n  \text{\enspace \enspace ($c_5 = c_4 - c_3$)} \displaybreak[0]\\
\end{align*}

\noindent
Note that through our choice of $c_0$ and $c_3$, we can ensure that
$c_5$ is a positive constant.  If $\lg(n -\tau) = \lg n$, then the
induction step clearly holds.  The alternative case can only happen
when $n$ is greater than a power of 2, and $n-\tau$ is less than a
power of two, due to the ceiling function on $\lg$.  Thus, the
alternative case only occurs once every $O(n / \lg n)$ times we remove
a biclique, since each biclique contains $O(\lg n)$ elements.  By
charging this extra cost to the rightmost negative term, the induction
holds. \qed
\end{proof}

We now show how to support precedence queries on a $(\lg n)$-flat
poset.  As in the previous section, if element $a$ is removed in the
dense case, we say $a$ is \emph{associated} with the output on lines
6-9.  Similarly, for each antichain $U_i \in \tilde{\mathcal{U}}$
involved in a merge step as the upper antichain in the sparse case, we
say that $U$ is \emph{associated} with the output of that merge step,
and vice versa.

\begin{lemma}\label{lem:flat-index}
Let $\tilde{P}$ be a $(\lg n)$-flat poset on $n$ elements, with
antichain decomposition $\tilde{U} = \{U_1, ..., U_m\}$.  There is a
data structure of size $n^2 / 4 + O((n^2 \lg \lg n) / \lg n)$ bits
that, given two elements $a$ and $b$, can report whether $a$ precedes
$b$ in $O(1)$ time.
\end{lemma}

\begin{proof}[Sketch]
We augment the output of $\textsc{Compress-Flat}$ with additional data
structures in order to answer queries efficiently.  Let $D_0$ be an
empty set. We denote the first set of elements removed in a dense case
as $D_1$, the second set as $D_2$ and so on.  Let $D_r$ denote the
last set of elements removed in a dense case, for some $r = O(n\lg \lg
n/\lg n)$.  Let $S_\ell = S/(\cup_{i = 0}^{\ell-1}D_i)$, for $1 \le
\ell \le r+1$.  We define $M_\ell(x)$ to be the number of elements $a
\in S_\ell$ such that $S(y) = a$, and $y \le x$. We now discuss how to
compute $M_\ell(x)$ in $O(1)$ time using a data structure of size
$O(n^2 \lg \lg n/ \lg n)$ bits.  Define $M'_\ell$ to be a bit
sequence, where $M'_\ell[x] = 1$ iff $S(x) \in S_{\ell}$, for $x \in
[1,n]$.  We represent $M'_{\ell}$ using the data structure of
Lemma~\ref{lem:rrr}, for $1 \le \ell \le r + 1$.  Overall, these data
structures occupy $O(n^2 \lg \lg n/ \lg n)$ bits, since $r = O((n \lg
\lg n) / \lg n)$, and each binary string occupies $O(n)$ bits by
Lemma~\ref{lem:rrr}.  To compute $M_\ell(x)$ we return
$\rankop_1(M'_\ell,x)$, which requires $O(1)$ time by
Lemma~\ref{lem:rrr}.  By combining the index just described with
techniques similar in spirit to those used in
Lemma~\ref{lem:flatten-index}, we can support precedence queries in
$O(1)$ time.  The idea is to find the output associated with the query
elements, and find the correct bit in the output to examine using the
index just described; the details can be found in
Appendix~\ref{app:flat-index}.
\end{proof}

Theorem~\ref{thm:main} follows by combining
Lemmas~\ref{lem:flatten-index} (with $\gamma$ set to $\lg n$)
and~\ref{lem:flat-index}.

\section{\label{sec:remark}Concluding remarks}

In this paper we have presented the first succinct data structure for
arbitrary posets.  For a poset of $n$ elements, our data structure
occupies $n^2/4 + o(n^2)$ bits and can support precedence queries in
$O(1)$ time.  This is equivalent to supporting $O(1)$ time queries of
the form, ``Is the edge $(a,b)$ in the transitive closure graph of
$P$?''  

Our first remark is that if we want to support edge queries on the
transitive reduction instead of the closure, a slightly simpler data
structure can be used.  The reason for this simplification is that for
the transitive reduction, our key lemma does not require the
antichains containing the biclique to be consecutive, and,
furthermore, we can ``flatten'' the transitive reduction in a much
simpler way than by using Lemma~\ref{lem:flatten-index}.  We defer
additional details to the full version. Our second remark is that, in
terms of practical behaviour, there are alternative representations of
bit sequences that support our required operations efficiently (though
not $O(1)$ time), and have smaller lower order terms in their space
bound (e.g.,~\cite{OS07}).  In practice, using these structures would
reduce the lower order terms significantly.  Finally, we remark that
we can report the neighbours of an arbitrary element in the transitive
closure graph efficiently, without asymptotically increasing the space
bound of Theorem~\ref{thm:main}.  This is done by encoding the
neighbours using a bit sequence, if there are few of them, and
checking all $n-1$ possibilities via queries to the data structure of
Theorem~\ref{thm:main}, if there are many. We defer the details until
the full version.

\bibliographystyle{plain} 
\bibliography{poset}

\newpage

\appendix
\section{Generalization to Transitive Binary Relations\label{app:transitive}}

In this section we discuss how to generalize Theorem~\ref{thm:main} to
transitive binary relations.  We make use of some notation described
in Section~\ref{sec:ds}, so we recommend reading that section first.

\begin{theorem}\label{thm:transitive-relation}
Let $T = (S, \preceq)$ be a transitive binary relation $\preceq$ on a
set of elements $S$, where $|S| = n$.  There is a succinct data
structure for representing $T$ that occupies $n^2/4 + O(( n^2 \lg \lg
n) / \lg n)$ bits, and can support precedence queries in $O(1)$ time:
i.e., given two elements $a, b \in S$, report whether $a \preceq b$.
\end{theorem}

\begin{proof}
Given a transitive binary relation, $T = (S,\preceq)$, we store a bit
sequence $B$, where $B[i] =1$ iff $S(i) \preceq S(i)$.  Thus, by using
$n$ bits, we can report whether $a \preceq a$ in $O(1)$ time, for any
$a \in S$.  At this point, we define a quasiorder $Q = (S,\preceq')$,
where $a \preceq' b$ iff $a \preceq b$, for all \emph{distinct
  elements} $a,b \in S$.  We represent the $Q$ using
Corollary~\ref{cor:digraph}.  Given $a,b \in S$, if $a = b$, and $S(i)
= a$, then we query $B$ and report ``yes'' iff $B[i] = 1$, otherwise,
we query the representation of $Q$ to determine whether $a$ precedes
$b$. \qed
\end{proof}

\section{\label{app:inequality}Proof of inequality used in Lemma~\ref{lem:flatten}}
The inequality is proved by induction on $t$, fixing $s = 1$ (since
the actual value of $s$ is irrelevant).  Base case: $t = 2$ holds
since $(n_1 + n_2)(n_1 + n_2 - 1)/2 \ge n_1 n_2$ for all integers
$n_1, n_2 \ge 1$.  Inductive step: Assume the inequality holds for all
$2 \le t_0 < t$.  We have:
\begin{align*}
  &\enspace \sum_{i = 1}^{t-1} \left( \left( \sum_{j = 1}^{i} n_j \right) n_{i+1} \right) = \sum_{i = 1}^{t-2}\left( \left( \sum_{j = 1}^{i} n_j \right) n_{i+1} \right) + \left(\sum_{j=1}^{t-1} n_j\right)n_t \\
= &\enspace n_{1,t-1}\left(\frac{n_{1,t-1} - 1}{2} + n_t\right) \\
= &\enspace \left(n_{1,t} - n_t\right)\left(\frac{n_{1,t} - 1 + n_t}{2} \right) \\
= &\enspace n_{1,t}\left(\frac{n_{1,t} - 1 + n_t}{2} \right) - n_{t}\left(\frac{n_{1,t} - 1 + n_t}{2} \right) \\
= &\enspace n_{1,t}\left(\frac{n_{1,t} - 1}{2} \right) + \frac{n_{1,t}n_t}{2} - \frac{n_{t}n_{1,t}}{2} + \frac{n_t}{2} - \frac{n_t^2}{2} \\
= &\enspace n_{1,t}\left(\frac{n_{1,t} - 1}{2} \right) - \frac{n_t(n_t - 1)}{2} \\
\le &\enspace n_{1,t}\left(\frac{n_{1,t} - 1}{2} \right) 
\end{align*}
\noindent
Which completes the proof.

\section{\label{app:compress-flat}Extra Details for Lemma~\ref{lem:compress-flat}}

In order to achieve $O((n^2 \lg \lg n)/\lg n)$ bits for the sparse
case, we need to use the standard trick in succinct data structures of
\emph{concatenating} all of the bit sequences output during the merge
steps into one long bit sequence, before applying Lemma~\ref{lem:rrr}
to the sequence.  Note that we can still perform rank operations on an
arbitrary range $[x_1,x_2]$ of this concatenated sequence, by
adjusting our search to take into account the number of 1s in the
prefix $[1,x_1-1]$.  Since this can be computed using a single rank
operation, it does not affect the time required to perform rank
operations.  By storing this concatenated sequence in the data
structure of Lemma~\ref{lem:rrr}, we guarantee that the lower order
term in the space bound will not dominate the space bound.  By the
same analysis presented in Lemma~\ref{lem:flatten}, the length of the
concatenated bit sequence will be $O(n^2)$ bits.  Thus, the size of
the lower order terms will be $O((n^2 \lg \lg n)/\lg n)$ bits.

\section{\label{app:flat-index}Proof of Lemma~\ref{lem:flat-index}}
We augment the output of $\textsc{Compress-Flat}$ with additional data
structures in order to answer queries efficiently.  Let $D_0$ be an
empty set. We denote the first set of elements removed in a dense case
as $D_1$, the second set as $D_2$ and so on.  Let $D_r$ denote the
last set of elements removed in a dense case, for some $r = O(n\lg \lg
n/\lg n)$.  Let $S_\ell = S/(\cup_{i = 0}^{\ell-1}D_i)$, for $1 \le
\ell \le r+1$.  We define $M_\ell(x)$ to be the number of elements $a
\in S_\ell$ such that $S(y) = a$, and $y \le x$. We now discuss how to
compute $M_\ell(x)$ it in $O(1)$ time using a data structure of size
$O(n^2 \lg \lg n/ \lg n)$ bits.  Define $M'_\ell$ to be a bit
sequence, where $M'_\ell[x] = 1$ iff $S(x) \in S_{\ell}$, for $x \in
[1,n]$.  We represent $M'_{\ell}$ using the data structure of
Lemma~\ref{lem:rrr}, for $1 \le \ell \le r + 1$.  Overall, these data
structures occupy $O(n^2 \lg \lg n/ \lg n)$ bits, since $r = O((n \lg
\lg n) / \lg n)$ bits, and each binary string occupies $O(n)$ bits by
Lemma~\ref{lem:rrr}.  To compute $M_\ell(x)$ we return
$\rankop_1(M'_\ell,x)$, which requires $O(1)$ time by
Lemma~\ref{lem:rrr}.

Consider an element $a$ removed during the dense case as part of the
biclique $D_k$.  When we refer to $a$ we will often reference the
antichains $\hat{U}_i$ and $\hat{U}_{i+1}$ such that $D_k \subset
\hat{U}_i \cup \hat{U}_{i+1}$ (see line 6).  Note that the indices $i$
and $i+1$ do \emph{not necessarily} correspond to the indices of
antichains in the initial antichain decomposition, $\tilde{U}$.  We
store an array $C$, where:

\begin{itemize}
\item $C[a].\texttt{id}$ is the value $k$ such that $a \in D_k$, or
  $\infty$ if $a$ was not removed;
\item $C[a].\texttt{rank}$ is the value $x$ such that $D_k(x) = a$;
\item $C[a].\texttt{top}$ is a bit indicating whether $a$ was in
  $U_{i+1}$, when $D_k$ was removed;
\item $C[a].\texttt{pnt}$ is a pointer to the output associated with
  $a$, $W^{-}_a$, $W^{+}_a$, and $Y$;
\item $C[a].\texttt{ds}$ the number of elements with rank less than
  $a$ in $\hat{U}_i \cup \hat{U}_{i+1}$;
\item $C[a].\texttt{dt}$ the number of elements with rank greater than
  $a$ in $\hat{U}_i \cup \hat{U}_{i+1}$.
\end{itemize}

Similar in spirit to Lemma~\ref{lem:flatten-index}, we store an $O(n
\lg n)$ bit array that in $O(1)$ time, for elements $a$ and $b$
returns $i,j,x$ and $y$ such $U_i,U_j \in \tilde{\mathcal{U}}$,
$U_i(x) = a$, and $U_j(y) = b$.  Note that in this case, the indices
\emph{do} correspond to the indices of the antichains in the initial
antichain decomposition $\tilde{\mathcal{U}}$.  We also store an array
$A$ of records, where, for each antichain $U_j \in
\tilde{\mathcal{U}}$, if $U_j$ was the upper antichain in a merge step
during a sparse case:
\begin{itemize}
\item $A[j].\texttt{pnt}$ points to the beginning of the sequence,
  $B$, associated with $U_j$, or null if no sequence is associated
  with $U_j$;
\item $A[j].\texttt{delta}$ stores the value $\ell$ such that the
  merge step occurred after the element set $D_{\ell - 1}$ was removed,
  and before $D_\ell$ was removed.
\end{itemize}

Finally, we store an array of partial sums $F$, where $F[i] = \sum_{k
  = 1}^{i-1} |U_k|$.  All these additional data structures occupy
$O((n^2\lg \lg n )/ \lg n)$ bits, so the claimed space bound holds by
Lemma~\ref{lem:compress-flat}.

\paragraph{Query Algorithm:}
If $i = j$, then we return ''no''. Otherwise, we assume, without loss
of generality, $i > j$. There are several cases:

\begin{enumerate}

\item \label{case:a} If $C[a].\texttt{id} = C[b].\texttt{id}$ and
  $C[a].\texttt{id} \neq \infty$, then:

\begin{enumerate}

\item If $C[a].\texttt{top} \neq C[b].\texttt{top}$, then report
  ``yes'', since there must be an edge between $a$ and $b$ in the
  removed biclique.

\item \label{case:ms}Otherwise, use $A[i].\texttt{pnt}$ to locate the
  bit sequence $B$, let $\ell = A[i].\texttt{delta}$, and $z =
  M_\ell(F[j] + y)$. We report ``yes'' if $B^z[M_\ell(F[i] + x)] = 1$
  and ``no'' otherwise.

\end{enumerate}

\item If $C[a].\texttt{id} = C[b].\texttt{id} = \infty$, then the
  procedure is similar to case~\ref{case:ms}.

\item If $C[a].\texttt{id} < C[b].\texttt{id}$,
  then let $\ell = C[a].\texttt{id}$.

\begin{enumerate} 
\item If $A[i].\texttt{top} = 1$ and $M_\ell(F[i] + x) -
  C[a].\texttt{ds} \le M_\ell(F[j] + y)$, then consider the binary
  string $W^{+}_a$, that we can locate using $C[a].\texttt{pnt}$.  If
  $A[j].\texttt{delta} > \ell$, then bit $W^{+}_a[M_\ell(F[j] + y) -
    M_\ell(F[j])]$ indicates whether there is an edge from $a$ to $b$.
  Otherwise, we check bit $W^{+}_a[M_\ell(F[j] + y)]$.

\item If $A[i].\texttt{top} = 0$ and $M_\ell(F[i] + x) - C[a].ds -1
  = 0$, then the bit we want to examine was output during a merge
  step, and we handle this as in case~\ref{case:ms}.

\item Otherwise, consider the sequence of integers, $Y$, that we can
  locate using $C[a].\texttt{pnt}$.  By examining $Y[M_\ell(F[j] +
    y)]$ and $C[a].\texttt{rank}$ we can determine whether there is an
  edge from $a$ to $b$ in $O(1)$ time\footnote{Briefly, we can use
    word-level parallelism, since $Y[M_\ell(F[j] + y)]$ fits in $O(1)$
    words.}.
\end{enumerate}

\item If $C[b].\texttt{id} < C[a].\texttt{id}$, then let $\ell =
  C[b].\texttt{id}$.
\begin{enumerate}
\item If $A[i].\texttt{delta} \le \ell$, then the bit we want to
  examine was output during a merge step, and we handle this as in
  case~\ref{case:ms}.

\item If $B[i].\texttt{top} = 0$, and $M_\ell(F[j] + y) +
  C[b].\texttt{dt} \ge M_\ell(F[i] + y)$, then consider the binary
  string $W^{-}_b$, that we can locate using $C[b].\texttt{pnt}$.  Let
  $z = M_\ell(F[i]+x) - M_\ell(F[i])$.  The bit $W^{-}_b[z]$ indicates
  whether there is an edge from $a$ to $b$.
\item Otherwise, consider the sequence of integers $Y$, that we can
  locate using $C[b].\texttt{pnt}$.  We examine $Y[M_\ell(F[i] + x) -
    C[b].\texttt{ds} - C[b].\texttt{dt} - 1]$ and $C[b].\texttt{rank}$
  to determine whether $a$ is connected to $b$.  Notice that we must
  correct for the fact that the two consecutive antichains,
  $\hat{U}_i$ and $\hat{U}_{i+1}$, that contain $D_\ell$ are not part
  of the set $H$ on line 9.
\end{enumerate}  
\end{enumerate}
\qed

\end{document}